\documentclass[runningheads]{llncs} 

\usepackage{booktabs} 
\usepackage{paralist} 
\usepackage{verbatim} 
\usepackage{subfig} 
\usepackage{amsmath}
\usepackage{amsfonts}
\usepackage{amssymb}

\usepackage{amsthm}
\usepackage{bbm}
\usepackage{bbold}
\usepackage{mathtools}
\usepackage{algorithm}
\usepackage{algpseudocode}
\usepackage[inline]{enumitem}
\usepackage{sgamevar}
\usepackage{color} 
\usepackage[hidelinks]{hyperref}
\usepackage{tikz}
\usetikzlibrary{intersections,calc,arrows.meta}
\let\ALP \mathcal 

\newcommand{\mc}{\mathcal}

\newcommand{\colspan}[1]{\texttt{ColSpan}(#1)}

\newcommand{\rank}[1]{\texttt{rank}({#1})}
\newcommand{\transpose}{\mathsf{T}}
\renewcommand{\vec}[1]{\boldsymbol{\mathbf{#1}}} 
\renewcommand{\Re}{\mathbb{R}}
\newcommand{\vct}[1]{\boldsymbol{\mathbf{#1}}} 

\newcommand{\Q}{\mathbb{Q}}

\newcommand{\bigO}{\mathcal{O}}

\newcommand{\ene}[1]{\vec{\bar{#1}}}

\DeclarePairedDelimiter\abs{\lvert}{\rvert}%
\DeclarePairedDelimiter\norm{\lVert}{\rVert}%

\makeatletter
\let\oldabs\abs
\def\abs{\@ifstar{\oldabs}{\oldabs*}}
\let\oldnorm\norm
\def\norm{\@ifstar{\oldnorm}{\oldnorm*}}
\makeatother
\begin{document}
\title{Two Algorithms for Computing Exact and Approximate Nash Equilibria in Bimatrix Games}
\titlerunning{Two Algorithms for Computing Exact and Approximate NE}

 \author{Jianzong Pi\inst{1}\orcidID{0000-0003-4603-7805} \and
 Joseph L. Heyman\inst{2}\orcidID{0000-0002-1898-7523}\thanks{The second author was fully supported by the United States Military Academy and the Army Advanced Civil Schooling (ACS) program. The views expressed in this work are those of the authors and do not reflect the official policy or position of the Department of the Army, Department of Defense, or the U.S. Government} \and
Abhishek Gupta\inst{1}\orcidID{0000-0003-1117-325X}\thanks{The third author gratefully acknowledges support from NSF Grant 1565487.}}
    \authorrunning{J. Pi \and J. L. Heyman \and A. Gupta}
    \institute{The Ohio State University, Columbus OH 43210, USA
    \email{\{pi.35,gupta.706\}@osu.edu} \and 
    United States Military Academy, West Point NY 10996, USA
    \email{joseph.heyman@westpoint.edu}}
 %

\begin{titlepage}

\maketitle
\begin{abstract}
In this paper, we first devise two algorithms to determine whether or not a bimatrix game has a strategically equivalent zero-sum game. If so, we propose an algorithm that computes the strategically equivalent zero-sum game. If a given bimatrix game is not strategically equivalent to a zero-sum game, we then propose an approach to compute a zero-sum game whose saddle-point equilibrium can be mapped to a well-supported approximate Nash equilibrium of the original game. We conduct extensive numerical simulation to establish the efficacy of the two algorithms. 
\end{abstract}

\end{titlepage}

\section{Introduction}
Non-cooperative game theory has become a popular method for modeling strategic interactions between decision makers, commonly referred to as ``players''. While originally applied to strategic interactions in games and economics, game theory is now gaining popularity in the fields of social and political sciences. Today, it is not uncommon to find game theory applied in diverse fields such as traffic engineering \cite{jiang2009cooperative,gong2013reducing,zhao2013load}, online advertising  \cite{dalessandro2012causally,choi2019learning}, cyber-security \cite{shiva2010game,do2017game}, and many others.

While game theory has been a widely used modeling technique, efficient computation of a solution to a game has been difficult. In 1950, John Nash defined Nash equilibrium (NE) as a solution concept in non-cooperative games, in which each player's expected payoff is maximized with the knowledge of other players' strategies. Moreover, each player will receive less expected payoff if he/she deviates from the NE. Nash proved that there exists an equilibrium in every finite-action finite-player non-cooperative game \cite{nash1951} using Brouwer's Fixed Point Theorem. However, there is currently no known efficient algorithm for computing NE in general non zero-sum games, and it has remained an open problem for over 70 years. In a series of works, \cite{daskalakis2005three,daskalakis2009complexity}, Daskalakis et al. showed that finite games with more than three players is Polynomial Parity Arguments on Directed graphs (PPAD)-complete. While the computation of NE may seem simpler in two player games, Chen et al. showed the hardness results for this case a short time later \cite{chen2009}.

With these hardness results established, one could hope to come up with an efficient algorithm to compute an approximate NE, which is also referred to as an $\varepsilon$-NE in two player games. For an $\varepsilon$-NE solution, any deviations from the $\varepsilon$-NE can gain either player at most an additional payoff of $\varepsilon$. In \cite{lipton2003}, Lipton et al. proposed a quasipolynomial time algorithm for computing an approximate NE for any fixed $\varepsilon$. The current ``best" polynomial time algorithm for a fixed $\varepsilon$ is due to \cite{tsaknakis2008optimization}, where Tsaknakis and Spirakis proposed an algorithm for $\varepsilon = 0.3393$. However, computation of $\varepsilon$-NE is still PPAD-complete if $\varepsilon$ is inversely polynomial in the size of the game \cite{chen2006computing}.

A more demanding notion of an approximate solution is the $\varepsilon$-well supported approximate Nash
equilibrium ($\varepsilon$-WSNE) \cite{daskalakis2009complexity}. In an $\varepsilon$-WSNE, players only place positive probability on strategies that have a payoff within $\varepsilon$ of the pure best response. $\varepsilon$-WSNE is a more restrictive approximation than $\varepsilon$-NE, as every $\varepsilon$-WSNE is an $\varepsilon$-NE, while the converse is not true \cite{daskalakis2009complexity}.

There is significantly less literature studying $\varepsilon$-WSNE compared to $\varepsilon$-NE. For the case of $\varepsilon$-WSNE, the first and most well-known polynomial time algorithm for a fixed $\varepsilon$ is $\frac{2}{3}$-WSNE as published in \cite{kontogiannis2010well}. Fearney et al. \cite{fearnley2016approximate} made  improvement to the previous algorithm, which resulted in $\varepsilon=\frac{2}{3}-0.005913759$. For random bimatrix games, Panagopoulo and Spirakis \cite{panagopoulou2014random} found that the uniform mixed strategy profile is, with high probability, a $\sqrt{\frac{3\ln{n}}{n}}$-WSNE.

In our work, we first devise an efficient algorithm for computing strategically equivalent zero-sum games using simple algebraic manipulations. Then, we propose a polynomial time algorithm for computing an $\varepsilon$-WSNE. To determine the latter algorithm, we define a certain affine game transformation to that leads to a simple linear program that outputs a zero-sum game that is ``close" to the original nonzero-sum game. We show that any NE of this zero-sum game is an $\varepsilon$-WSNE of the original game.

\subsection{Notation}
In this paper, all vectors are column vectors and are written in bold font. We denote $\vec{1}_n$ and $\vec{0}_n$ as the ones and zeros vector of length $n$ and denote $e_j$, $j\in\{1,2,...,n\}$, as a vector with 1 at the $j$\ts{th} position and $0$'s elsewhere. We denote $\Delta_n \subset \mathbb{R}^n$ as the set of probability distributions over $\{1,..., n\}$, i.e., $\Delta_{n}=\big\{\mathbf{p} \mid p_i\geq0,\forall i\in \{1,\dots,n\},\sum_{i=1}^n p_i = 1  \big\}$. For a matrix $A \in \mathbb{R}^{m\times n}$, we use $a_{i,j}$ to denote the entry on $i$-th row and $j$-th column of $A$. Moreover, we define $A_{(i)}$ as the $i$-th row of $A$, and $A^{(j)}$ as the $j$-th column of $A$. We define the max norm of matrix $A$ as $\|A\|_{\max} = \max_{i,j}|a_{i,j}|$. Define $\colspan{A}$ as the subspace spanned by the columns of matrix $A$. 
We define 
$\ALP D_n \subseteq \mathbb{R}^{n\times n}$
as the set of all diagonal matrices with $n$ rows and $n$ columns with positive diagonal entries. For a diagonal matrix $D \in \ALP D_n$, we use $d_{j}$ to denote the $(j,j)$-th element in matrix $D$ and note that $d_j>0$ for all $j\in\{1,\ldots,n\}$.

\subsection{Outline of the paper}
In Section \ref{sec:prelim}, we begin by introducing some preliminary concepts related to this paper. In Section \ref{sec:sezsg}, we propose necessary and sufficient conditions of existence of a strategically equivalent zero-sum game, then we devise an algorithm that efficiently computes the strategically equivalent zero-sum game and provide a complexity analysis. We will also run an experiment evaluate the performance. Then, in Section \ref{sec:affine}, we propose another efficient algorithm to compute approximate Nash equilibrium if the conditions in Section \ref{sec:sezsg} are not satisfied. The auxillary theorems for the proof in Section \ref{sec:sezsg} are provided in Appendix \ref{sec:auxillary}.
\section{Preliminaries}\label{sec:prelim}
In this section, we recall the definitions of Nash equilibrium and approximate Nash equilibrium. We focus on bimatrix games (2-player games) in this paper. Every bimatrix game can be defined by a tuple $(m,n,A,B)$, where player 1 has $m$ actions, player 2 has $n$ actions, and $A, B\in \mathbb{R}^{m \times n}$ are the payoff matrices of player 1 and 2. Both players can choose to use pure strategies, that is, they can choose a single action from its own set of pure strategies denoted by $S_1 = \{1, ..., m\}$ and $S_2 = \{1,..., n\}$. If the players play pure strategies $(i, j)\in S_1\times S_2$, player 1 and player 2 will receive payoffs $a_{i,j}$ and $b_{i,j}$, respectively. Players may also play mixed strategies in bimatrix games. Player 1 and player 2 can choose probability distributions $\vec p\in\Delta_m$ over $S_1$ and $\vec q\in\Delta_n$ over $S_2$. Then, player 1 has expected payoff $\vec p^\transpose A\vec q$ and player 2 has expected payoff $\vec p^\transpose B\vec q$. 

We state the definition of best response condition and best response correspondence in order to define Nash equilibrium. The best response is the mixed strategy that gives the best outcome for one player, given the mixed strategy of the other player. This is made precise in the following definition.
\begin{definition}[Best response condition \cite{nash1951}]\label{def:bestResponse}
 Let $\vec{p}$ and $\vec{q}$ be mixed strategies of player 1 and player 2. Then $\vec{p}$ is a best response to $\vec{q}$ if and only if for all $i\in S_1$,
 \begin{equation*}
     p_i>0 \implies (A\vec{q})_i=u=\max_{k\in S_1} (A\vec{q})_k,
 \end{equation*}
 and $\vec{q}$ is a best response to $\vec{p}$ if and only if for all $j\in S_2$,
 \begin{equation*}
     q_j>0 \implies (B^\transpose \vec{p})_j=v=\max_{k\in S_2} (B^\transpose \vec{p})_k.
 \end{equation*}
\end{definition}

\begin{definition}[Best response correspondence]
    For the payoff matrices $A$ and $B$, define the best response correspondences $\Gamma_A:\Delta_n\rightrightarrows \Delta_m$ and $\Gamma_B:\Delta_m\rightrightarrows \Delta_n$ as
\begin{align*}
 \Gamma_A(\vec q) &= \{\vec p\in\Delta_m: \vec p^\transpose A \vec q = \max_i [A\vec q]_i \},\\
 \Gamma_B(\vec p) &= \{\vec q\in\Delta_n: \vec p^\transpose B \vec q = \max_j [\vec p^\transpose B]_j \}.
\end{align*}
\end{definition}

When each player's mixed strategy is a best response to the other player's strategy, their strategies form a Nash Equilibrium. This is made precise in the following definition.

\begin{definition}[Nash equilibrium \cite{nash1951}]
    A pair $(\vec p, \vec q)$ of mixed strategies is a Nash Equilibrium (NE) if and only if $\vec p \in \Gamma_A(\vec q)$ and $\vec q \in \Gamma_B(\vec p)$.
\end{definition}

In zero-sum games, we usually call the Nash equilibrium between two players as a saddle point equilibrium (SPE).

\subsection{Strategically Equivalent Games}\label{sub:stratEq}

We focus on a bimatrix game $(m,n,A,B)$ as defined above. We define the Nash equilibrium correspondence as $\Phi(A, B): \mathbb R ^{m\times n}\times \mathbb R ^{m\times n} \rightrightarrows \Delta_m \times \Delta_n$. It was proved in \cite{nash1951} that every bimatrix game with a finite set of pure strategies has at least one NE in mixed strategies. Thus, the image $\Phi(A, B)$ is nonempty for any $(A, B) \in \mathbb R ^{m\times n}\times \mathbb R ^{m\times n}$. We say two games are strategically equivalent when both games have the same set of players, the same set of strategies, and the same set of NE. An equivalent definition of strategic equivalence based on the preference ordering for all mixed strategies $\vec p \in \Delta_m$ and $\vec q \in \Delta_n$ was given by Moulin and Vial \cite{moulin1978strategically}, which is stated below.

\begin{definition}[Strategically equivalence]\label{def:strateEq}
Two bimatrix games $(m,n,A,B)$ and $(m,n,\bar A,\bar B)$ are strategically equivalent if and only if for any $\vec{\bar p}, \vec p \in \Delta_m$ and $\vec{\bar q}, \vec q \in \Delta_n$, we have
\begin{align*}
    \vec{\bar p}^TA\vec q \geq \vec p^T A \vec q \iff \vec{\bar p}^T\bar A\vec q \geq \vec p^T \bar A \vec q, \\
    \vec p^TB\vec{\bar q} \geq \vec p^T B\vec q \iff \vec p^T\bar B\vec {\bar q} \geq \vec p^T \bar B \vec q. 
\end{align*}
\end{definition}

Proving that two games are strategically equivalent in Definition \ref{def:strateEq} is difficult since all $\vec{\bar p}, \vec p \in \Delta_m$ and $\vec{\bar q}, \vec q \in \Delta_n$ need to be checked to satisfy the above conditions. However, there are several classes of transformations such that strategically equivalence is naturally conserved. Positive affine transformations (PAT) ensure the strategically equivalence of two games; that is, if two games have a PAT correspondence, then they are strategically equivalent. We define the PAT correspondence $\Upsilon$ as follows:

\begin{definition}[PAT Correspondence \cite{heyman2019rank}]\label{def:PATCorrespondence}
The game $(m, n, \bar A, \bar B)$ is a PAT of $(m, n, A, B)$ if and only if there exists $\alpha_1, \alpha_2 \in \mathbb R_{>0}$, $\vec u\in \mathbb R^n$, and $\vec v \in \mathbb R^m$ such that $\bar A = \alpha_1 A + \vec 1_m\vec u^T$ and $\bar B = \alpha_2 B + \vec v \vec 1_n^T$. The map $\Upsilon:\mathbb R^{m\times n} \times \mathbb R^{m\times n} \rightrightarrows\mathbb R^{m\times n} \times \mathbb R^{m\times n}$ is a PAT correspondence if
\begin{align*}
    \Upsilon(A,B) = \big\{ (\bar A, \bar B) \in \mathbb R^{m\times n} \times \mathbb R^{m\times n}: (\bar A, \bar B) \text{ is a PAT of } (A, B) \big\}.
\end{align*}
\end{definition}

It is obvious that PAT preserves the preference orderings given in Definition \ref{def:strateEq}. Thus, two games are strategically equivalent if the two games have a PAT correspondence. Moreover, \cite{moulin1978strategically} also showed that the converse also holds. With the results above, we have the following lemma.

\begin{lemma}[\cite{moulin1978strategically}]\label{lem:stratEqVec}
	Two games $(m, n, A, B)$ and $(m, n, \bar A, \bar B)$ are strategically equivalent if and only if $(\bar A, \bar B)\in \Upsilon(A, B)$.
\end{lemma}

We say two games are strategically equivalent via a PAT if they have a PAT correspondence. In a special case where $\alpha_1=\alpha_2=1$, we say that those two games are strategically equivalent via a 1-PAT.

\subsection{Approximate Nash Equilibrium}

While \cite{nash1951} showed that NE exists in all finite games, it remains an open problem to find an algorithm to compute NE efficiently in general bimatrix games \cite{daskalakis2009complexity}. On the other hand, there exists some efficient algorithms to compute approximate NE.

\begin{definition}[$\varepsilon$-well-supported Nash Equilibrium]\label{def:WSNE}
	We refer to the pair of strategies $(\vec{\tilde{p}},\vec{\tilde{q}})$ as an  epsilon-well-supported Nash Equilibrium ($\varepsilon$-WSNE) of game $(m, n,A,B)$ if and only if:
 \begin{align*}
     \text{for all } i\in S_1, k\in S_1, \; \tilde{p}_i>0 &\implies (A\vec{\tilde{q}})_i\geq ( A\vec{\tilde{q}})_k-\varepsilon, \;\\
     \text{for all } j\in S_2, l\in S_2\; \tilde{q}_j>0 &\implies (B^\transpose \vec{\tilde{p}})_j\geq (B^\transpose \vec{\tilde{p}})_l-\varepsilon.
 \end{align*}
\end{definition}

We also define the $\varepsilon$-approximation of payoff matrices as follows.
\begin{definition}\label{def:epsilonApproximation}
    For matrices $\tilde A, R\in\Re^{m\times n}$, $\tilde A$ is an $\varepsilon$-approximation of $R$ if $\tilde A=R+E$, where $\varepsilon \geq \|E\|_{max}$.
\end{definition}
Inspired by \cite[Theorem 1]{kontogiannis2007efficient}, we have the following lemma on the connection between the approximation of the payoff matrices in a game and approximate Nash equilibrium of that game.

\begin{lemma}\label{lem:2epsWSNE}
Given the game $(m,n,\tilde A,\tilde B)$, let $\tilde A$ be an $\tilde \varepsilon_1$-approximation of $R$ and $\tilde B$ be an $\tilde \varepsilon_2$-approximation of $C$. If $(\tilde{\vec p},\tilde {\vec q})$ is an NE of the game $(m,n,\tilde A,\tilde B)$, then $(\tilde{\vec p},\tilde{\vec q})$ is a $2\tilde \varepsilon$-WSNE of the game $(m,n,R,C)$, where $\tilde \varepsilon=\max\{\tilde \varepsilon_1,\tilde \varepsilon_2\}$.
\end{lemma}
\begin{proof}
Applying Definitions \ref{def:bestResponse} and \ref{def:epsilonApproximation}, let $\tilde A = R+E$, where $\|E\|_{max}\leq\tilde \varepsilon$. For player 1 we have that:
 \begin{align*}
     \text{For all } i,k\in S_1, \; \bar{p}_i>0 &\implies (\tilde A\vec{\bar{q}})_i\geq (\tilde A\vec{\bar{q}})_k,\\
    &\iff (R\vec{\bar{q}}+E\vec{\bar{q}})_i\geq (R\vec{\bar{q}}+E\vec{\bar{q}})_k,\\
    &\iff (R\vec{\bar{q}})_i\geq (R\vec{\bar{q}})_k+(E\vec{\bar{q}})_k-(E\vec{\bar{q}})_i.
 \end{align*}
 Since $\varepsilon_1\geq \max_{i,j}\abs{e_{i,j}}$ and $\ene{q}\in \Delta_m$,
\begin{equation*}
    \text{For all } i,k\in S_1, \; (E\vec{\bar{q}})_k-(E\vec{\bar{q}})_i\geq -2\varepsilon_1.
\end{equation*}
Therefore,
\begin{equation*}
   \text{For all } i,k\in S_1, \; \bar{p}_i>0 \implies (R\vec{\bar{q}})_i\geq (R\vec{\bar{q}})_k-2\varepsilon_1.
\end{equation*}
The proof for player 2 is similar and thus omitted.
\end{proof}

In \cite{kontogiannis2007efficient} the authors show that additive transformations have no effect on the set of WSNE. Formally, we have the following lemma.

\begin{lemma}\label{lem:stratEqWSNE_PAT}
	Consider the games $(m,n,A,B)$ and $(m,n,R,C)$ which are strategically equivalent via 1-PAT. The strategy pair $(\tilde {\vec p},\tilde {\vec q})$ is an $\tilde \varepsilon$-WSNE of the game $(m,n,A,B)$ if and only if $(\tilde {\vec p},\tilde {\vec q})$ is an $\tilde \varepsilon$-WSNE of the game $(m,n,R,C)$.
\end{lemma}
\begin{proof}
    Suppose $(\tilde {\vec p},\tilde {\vec q})$ is an $\tilde \varepsilon$-WSNE of the game $(m,n,R,C)$. Applying Definition \ref{def:WSNE} for player 1, we have that:
     \begin{align*}
     \text{For all } i,k\in S_1, \; \tilde {\vec p}_i>0 &\implies (R\tilde{\vec q})_i\geq (R\tilde{\vec q})_k-\tilde \varepsilon,\\
    &\iff (A\tilde {\vec q})_i+(\vec{1}_m \vec{u}^\transpose\tilde{\vec q})_i\geq ( A\tilde {\vec q})_k+(\vec{1}_m \vec{u}^\transpose\tilde{\vec q})_k -\tilde \varepsilon,\\
    &\iff (A\tilde{\vec q})_i \geq (A\tilde{\vec q})_k-\tilde \varepsilon.
 \end{align*}
  The final step of this proof relies on the fact that $(\vec{1}_m \vec{u}^\transpose\tilde{\vec q})_i=(\vec{1}_m \vec{u}^\transpose\tilde{\vec q})_k$ for any $i,k\in S_1$.
 The proof for player 2 is similar and thus omitted. 
\end{proof}

In the next section, we will introduce an efficient algorithm to compute Nash equilibria for general bi-matrix games.
\section{A Fast Algorithm to Compute Strategically Equivalent Zero-Sum Games}\label{sec:sezsg}

In this section, we devise an algorithm that determines a strategically equivalent zero-sum game $(m,n,\bar A,\bar B)$ given a non-zero-sum game $(m,n,A,B)$ with $(\bar A,\bar B)\in\Upsilon(A,B)$ (and $\bar A+\bar B = 0$). This section is based on Chapter 3 of the second author's PhD dissertation \cite{heyman2019computation}.

We first introduce some notations. Define the set $\mathcal M_{m\times n} \subseteq \mathbb R^{m\times n}$ as\begin{align*}
    \mathcal M_{m\times n}(\mathbb R) &= \\
    &\left\{M\in \mathbb{R}^{m\times n}\mid \text{there exists } \vec u\in \mathbb R^n, \vec v \in \mathbb R^m \text{ s.t. }M = \vec 1_m\vec u^T + \vec v\vec 1_n^T\right\}.
\end{align*} 
Note that $\mathcal M_{m\times n}(\mathbb R)$ is a linear space over field $\mathbb R$, hence $\mathcal M_{m\times n}(\mathbb R)$ is closed under addition and scalar multiplication.
Moreover, for matrix $\bar A \in \mathbb R^{m\times n}$, we define the set $\mathcal{WZ}(\bar A)$ as
\begin{align*}
    \mathcal{WZ}(\bar A) \coloneqq \left \{ (\vec w, \vec z)\in \mathbb R^m \times \mathbb R^n \mid \vec w^\transpose \bar A \vec z \neq 0, \vec 1_m^\transpose \vec w = \vec 1_n^\transpose \vec z= 0\right \}.
\end{align*}
Then, given a nonzero-sum game $(m,n,A,B)$, the following theorem provides a necessary and sufficient condition to the existence of strategic equivalent zero sum game to the original game.

\begin{theorem}\label{thm:existzsg}
Consider a nonzero-sum game $(m,n,A,B)$, where $A,B\not\in\mathcal M_{m\times n}(\mathbb R)$. The nonzero-sum game $(m,n,A,B)$ is strategically equivalent to a zero-sum game $(m,n,\bar A, \bar B)$ (where $ \bar B = -\bar A$) if and only if the following conditions are satisfied: 
\begin{enumerate}
    \item[1.] For any $(\vec w, \vec z)\in \mathbb R^m \times \mathbb R^n$ such that $\vec{1}_m^\transpose \vec{w}=\vec{1}_n^\transpose \vec{z}=0$, and $\vec{w}^\transpose B\vec{z}\neq0$, \begin{equation*}
        \gamma \coloneqq -\frac{\vec w^TA\vec z}{\vec w^T B \vec z} > 0.
    \end{equation*}
    \item[2.] $M \coloneqq A + \gamma B \in \mathcal M_{m\times n}(\mathbb R)$.
\end{enumerate}
\end{theorem}

\begin{proof}
We start by proving the reverse direction. Suppose the two conditions are satisfied. Note that as $M = A + \gamma B \in \mathcal M_{m\times n}(\mathbb R)$, $\rank{M} \leq 2$. Next, we consider the three cases when $\rank{M}$ is $2$, $1$ and $0$.
\begin{enumerate}
    \item {\bf Case 1:} $\rank{M} = 2$. We can write $M$ as a summation of 2 rank-1 matrices, so $\vec 1_m \in \colspan{M}$ and $\vec 1_n \in \colspan{M^\transpose}$. Hence, there exists $\vec x_1 \in \mathbb R^n$ such that $M\vec x_1 = \vec 1_m $.
    Moreover, there exists $\vec y_1 \in \mathbb R^m$ such that $\vec y_1^\transpose M \neq \vec 0_m$ and $\vec 1_m^\transpose \vec y_1 \neq 0$. Let $w_1 = \vec y_1^\transpose M\vec x_1\neq 0$ and let $\hat {\vec u} := (w_1^{-1}\vec y_1^\transpose M)^\transpose$. Then, $w_1^{-1}M\vec x_1 \vec y_1^\transpose M = M - \vec 1_m\hat{\vec u}^\transpose$. By applying Wedderburn rank reduction formula \cite[p.69]{wedderburn1934lectures}, we have
    \begin{align*}
        M_2 = M -w_1^{-1}M\vec x_1 \vec y_1^\transpose M  = M - \vec 1_m\hat{\vec u}^\transpose.
    \end{align*}
    Next, we show the existence of $\vec {\hat v}\in \mathbb R^m$ such that $M_2 = \vec {\hat v} \vec 1_n^\transpose$. Indeed, $M \in \mathcal M_{m\times n}(\mathbb R)$ and $\vec 1_m\hat{\vec u}^\transpose \in \mathcal M_{m\times n}(\mathbb R)$. As $\mathcal M_{m\times n}(\mathbb R)$ is closed under addition, we have $M_2 \in \mathcal M_{m\times n}(\mathbb R)$. In addition, by Wedderburn \cite[p. 69]{wedderburn1934lectures}, $\rank{M_2}=1$. By Theorem \ref{thm:rankreduction}, $\vec 1_m \not \in \colspan{M_2}$ implies there exists $\hat{\vec v}\in \mathbb R^m$ such that $M_2 = \hat{\vec v}\vec 1_n^\transpose$. Finally, letting $\bar A := A - \vec 1_m \hat{\vec u}^\transpose$ and $\bar B := \gamma B - M_2 $, we have 
    \begin{align*}
        \bar A + \bar B =  A + \gamma B - \vec 1_m\hat{\vec u}^\transpose - M_2 = 0_{m\times n}.
    \end{align*}
    Since $\gamma>0$, we conclude that $(m, n, \bar A, \bar B)$ is strategically equivalent to $(m, n, A, B)$ via a PAT.
    \item {\bf Case 2:} $\rank{M} = 1$. 
    As $\rank{M} = 1$, either $\vec 1_m \in \colspan M$ or $\vec 1_n \in \colspan{M^\transpose}$.
    We first assume $\vec 1_m \in \colspan M$. Let $M = \vec 1_m \hat{\vec u}^\transpose$. Then we can let $\bar A = A - M$, and $\bar B = \gamma B$ and we have $ \bar A + \bar B =0_{m\times n}$. 
    
    The proof of the case $\vec 1_n \in \colspan{M^\transpose}$ is similar and therefore omitted. Hence, we have found $\bar A, \bar B$ such that $(m, n, A, B)$ is strategically equivalent to $(m, n, \bar A, \bar B)$ via a PAT. 
    \item {\bf Case 3:} $\rank{M} = 0$. Let $\bar A =  A$ and $\bar B = \gamma B$. Since $M = A + \gamma B = 0_{m\times n}$, $(m, n, \bar A, \bar B)$ is zero-sum and strategically equivalent to $(m, n, A, B)$ via a PAT.
\end{enumerate}
Next, we prove the forward direction and assume that the nonzero-sum game $(m,n,A,B)$ is strategically equivalent to a zero sum game $(m,n,\bar A, \bar B)$ via PAT. Since $B \not \in \mathcal M_{m\times n}(\mathbb R)$, Theorem \ref{thm:wznonempty} in Appendix \ref{sec:auxillary} implies $\mathcal{WZ}(B) \neq \emptyset$. Pick $(\vec w, \vec z)\in \mathcal{WZ}(B)$, we have the following
\begin{enumerate}
    \item[(a)] $    \vec w^\transpose A \vec z = \alpha_1 \vec w^\transpose \bar A \vec z + \vec w^\transpose\vec 1_m \vec u^\transpose \vec z = \alpha_1 \vec w^\transpose \bar A \vec z $,
    \item[(b)] $
    \vec w^\transpose B \vec z = -\alpha_2 \vec w^\transpose \bar A \vec z + \vec w^\transpose\vec v \vec 1_n^\transpose \vec z = -\alpha_2 \vec w^\transpose \bar A \vec z$.
\end{enumerate}
Since $(\vec w, \vec z) \in \mathcal{WZ}(B)$, $\vec w^\transpose B \vec z \neq 0$, hence $\gamma$ is well-defined. As $\alpha_1 >0$ and $\alpha_2 > 0$, we conclude that
\begin{align*}
    \gamma = -\frac{\vec w^T A\vec z}{\vec w^T  B \vec z} = \frac{\alpha_1\vec w^\transpose\bar A \vec z}{\alpha_2 \vec w^\transpose\bar  A\vec z} = \frac{\alpha_1}{\alpha_2}> 0.
\end{align*} 
Note it is straightforward to show that $\mathcal{WZ}(A) = \mathcal{WZ}(\bar A) =\mathcal{WZ}(\bar B)$. The fact that $M \coloneqq A + \gamma B$ is in $\mathcal M_{m\times n}(\mathbb R)$ follows from simple algebraic manipulations. The proof is thus complete.
\end{proof}
To turn the above theorem into a fast algorithm, we need to solve two specific problems. The first one is to determine whether or not the payoff matrices are in $\mathcal M_{m\times n}(\mathbb R)$ and the second one is to compute $(\vec w, \vec z) \in \mathcal{WZ}(B)$. In the next section, we derive two results that immediately yield linear time algorithms to solve these two problems.

\subsection{Algorithmic Implications for Matrices in $\mathcal M_{m\times n}(\mathbb R)$}
We first derive a fast approach to determine whether or not a matrix is in $\mathcal M_{m\times n}(\mathbb R)$ in the next theorem.
\begin{theorem}\label{thm:testFinM}
Given a matrix $F\in\Re^{m\times n}$, select any $(i,j)\in\{1\dots m\} \times \{1\dots n\}$ and let
\begin{align}\label{eq:prop:testFinM}
    \bar{F}&\coloneqq\vec{1}_mF_{(i)}+(F^{(j)}-\vec{1}_mf_{i,j})\vec{1}_n^\transpose.
\end{align}
Then $F$ is in $\mc{M}_{m\times n}(\Re)$ if and only if $F = \bar F$. 
\end{theorem}
\begin{proof}
To begin with, notice that $\vec 1_m F_{(i)} \in \mathcal{M}_{m\times n}(\mathbb R)$ and $(F^{(j)}-\vec{1}_mf_{i,j})\vec{1}_n^\transpose \in \mathcal{M}_{m\times n}(\mathbb R)$. 
Let $\hat F \coloneqq F - \bar F = F - \vec{1}_mF_{(i)}-(F^{(j)}-\vec{1}_mf_{i,j})\vec{1}_n^\transpose$, then $\hat F \in \mathcal{M}_{m\times n}(\mathbb R)$ if and only if $F \in \mathcal{M}_{m\times n}(\mathbb R)$, as $\mathcal{M}_{m\times n}(\mathbb R)$ is closed under addition. We will show that $\hat{F}$ is in $\mc{M}_{m\times n}(\Re)$ if and only if $\hat{F}=\vec{0}_{m\times n}$. We have that
\begin{align*}
    \hat{F}&=F-\vec{1}_mF_{(i)}-(F^{(j)}-\vec{1}_mf_{i,j})\vec{1}_n^\transpose\\
    &=\setlength\arraycolsep{1.2pt} \begin{bmatrix}
    f_{1,1}-f_{i,1}-f_{1,j}+f_{i,j}, & \dots & f_{1,n}-f_{i,n}-f_{1,j}+f_{i,j} \\
    f_{2,1}-f_{i,1}-f_{2,j}+f_{i,j}, & \dots & f_{2,n}-f_{i,n}-f_{2,j}+f_{i,j} \\
    \vdots                                                    & \ddots& \vdots                          \\
    f_{m,1}-f_{i,1}-f_{m,j}+f_{i,j},  & \dots & f_{m,n}-f_{i,n}-f_{m,j}+f_{i,j}
    \end{bmatrix}.
\end{align*}
In this form, it is clear that $\hat{F}_{(i)}=\vec{0}_n^\transpose$ and $\hat{F}^{(j)}=\vec{0}_m$. Since $\hat{F}_{(i)}=\vec{0}_n^\transpose$ and $\hat{F}^{(j)}=\vec{0}_m$, from Lemma \ref{lem:FinMis0} $\hat{F}$ is in $\mc{M}_{m\times n}(\Re)$ if and only if $\hat{F}=\vec{0}_{m\times n}$. 
\end{proof}

The matrix $\hat{F}$ constructed in Theorem \ref{thm:testFinM} depends on which indices $(i,j)$ were used to construct it, and the proper notation should be $\hat{F}(i,j)$. However, since in all of our results that require similar notation the selection of $(i,j)$ is arbitrary, in an abuse of notation we drop $(i,j)$ and simply use $\hat{F}$. An algorithm for determining if a matrix is in $\mc{M}_{m\times n}(\Q)$ is shown in Algorithm \ref{alg:FinM}.

We now turn our attention to an efficient manner to calculate $\vec{z}\in\big\{\vec{z}\in \Re^n \mid\vec{1}_n^\transpose \vec{z}=0\big\}$ and $\vec{w}\in\big\{\vec{w}\in\Re^m\mid \vec{1}_m^\transpose \vec{w}=0\big\}$ such that $(\vec w, \vec z) \in \mathcal{WZ}(F)$. This is used to compute a candidate parameter $\gamma$ in Theorem \ref{thm:existzsg}. The next corollary derives a computationally efficient method for computing $\vec{z}$ and $\vec{w}$ for $F\notin\mc{M}_{m\times n}(\Re)$. 

\begin{algorithm}[tbh]
	\caption{Algorithm for determining if a matrix is in $\mc{M}_{m\times n}(\Q)$}
	\label{alg:FinM}
	\begin{algorithmic}[1] 
		\Function{IsMatrixInM}{$F$}
        \State Select any $(i,j)\in\{1\dots m\} \times \{1 \dots n\}$
		\For{$s\gets1,m$}
		   \State $R_{(s)}\gets F_{(i)}$
		\EndFor
		\For{$t\gets1,n$}
		    \State $C^{(t)}\gets F^{(j)}-\vec{1}_m f_{i,j}$
		\EndFor
		\State $\hat{F}\gets F-R-C$
		\State MatrixInM$\gets1$
		\For{$l\gets1,m$} \label{state:alg:FinM:SearchF}
		    \For{$k\gets1,n$}
		        \If{$\hat{f}_{l,k}\neq0$}
		            \State MatrixInM$\gets0$
		            \State \textbf{break}
		         \EndIf
		    \EndFor
		    \If{MatrixInM$=0$}
		       \State \textbf{break}
		    \EndIf
		\EndFor \label{state:alg:FinM:SearchFend}
        \State \textbf{return} (MatrixInM)
		\EndFunction
	\end{algorithmic}
\end{algorithm}

\begin{corollary}\label{cor:prop:testFinM:A}
Given matrix $F\in\Re^{m\times n}$, $F\notin\mc{M}_{m\times n}(\Re)$, construct $\hat{F}$ as in Theorem \ref{thm:testFinM}. Choose any $(l,k)$ such that $\hat{f}_{l,k}\neq0$, let $\vec{w}=\vec{e}_l-\vec{e}_i$, and $\vec{z}=\vec{e}_k-\vec{e}_j$. Then $\vec{z}\in\big\{\vec{z}\in \Re^n \mid \vec{1}_n^\transpose \vec{z}=0\big\}$, $\vec{w}\in\big\{\vec{w}\in\Re^m\mid \vec{1}_m^\transpose \vec{w}=0\big\}$, and 
\begin{equation*}
\vec{w}^\transpose F \vec{z}=\vec{w}^\transpose \hat{F} \vec{z}=\hat{f}_{l,k}\neq0.
\end{equation*}
\end{corollary}
\begin{proof}
Clearly $\vec{w}$ and $\vec{z}$ are constructed such that $\vec{z}\in\big\{\vec{z}\in \Re^n \mid \vec{1}_n^\transpose \vec{z}=0\big\}$, $\vec{w}\in\big\{\vec{w}\in\Re^m\mid \vec{1}_m^\transpose \vec{w}=0\big\}$. From Theorem \ref{thm:testFinM} we have that:
\begin{align*}
    \hat{F}&= F-\vec{1}_mF_{(i)}-(F^{(j)}-\vec{1}_mf_{i,j})\vec{1}_n^\transpose.
\end{align*}
Then we have:
\begin{align*}
    \vec{w}^\transpose \hat{F} \vec{z}&=\vec{w}^\transpose F \vec{z}-\vec{w}^\transpose \vec{1}_mF_{(i)}\vec{z}-\vec{w}^\transpose (F^{(j)}-\vec{1}_mf_{i,j})\vec{1}_n^\transpose \vec{z} =\vec{w}^\transpose F \vec{z}.
\end{align*}
Since $\hat{F}$ was constructed such that $\hat{F}_{(i)}=\vec{0}_n^\transpose$ and $\hat{F}^{(j)}=\vec{0}_m$, we have that:
\begin{align*}
    \vec{w}^\transpose \hat{F} \vec{z}&=(\vec{e}_l-\vec{e}_i)^\transpose \hat{F} (\vec{e}_k-\vec{e}_j)\\
    &=(\vec{e}_l-\vec{e}_i)^\transpose(\hat{F}\vec{e}_k-\hat{F}\vec{e}_j)\\
    &=(\vec{e}_l-\vec{e}_i)^\transpose(\hat{F}\vec{e}_k-\vec{0}_m)\\
    &=\vec{e}_l^\transpose \hat{F}\vec{e}_k - \vec{e}_i^\transpose \hat{F}\vec{e}_k\\
    &=\vec{e}_l^\transpose \hat{F}\vec{e}_k - \vec{0}_n^\transpose \vec{e}_k\\
    &=\vec{e}_l^\transpose \hat{F}\vec{e}_k\\
    &=\hat{f}_{l,k}.
\end{align*}
Finally, $\hat{f}_{l,k}$ was selected such that $\hat{f}_{l,k}\neq0$. This completes the proof. 
\end{proof}

\subsection{A Simple Example: Rock-Paper-Scissors}\label{sub:example}

Consider the game matrix given in Figure \ref{fig:subRPS:RPS} and let us represent this game as $(m,n,A,-A)$. This is the classic Rock-Paper-Scissors (R-P-S) with well known NE strategies $\vec{p}^*=\vec{q}^*=[\frac{1}{3}, \frac{1}{3}, \frac{1}{3}]^\transpose$. The game in Figure \ref{fig:subPAT:RPS} is a positive affine transformation of $(m,n,C,-C)$. Let us represent this game as $(m,n,A,B)$. Clearly, $(m,n,A,B)$ is neither zero-sum nor constant-sum. However, by applying the process outlined above one can obtain the game in Figure \ref{fig:subSZSG:RPS}, which is a zero-sum game and strategically equivalent to $(m,n,A,B)$.

	\begin{figure}[htb]%
	\centering
		\subfloat[R-P-S]{
		\begin{game}{3}{3}
			\> R    \> P\> S\\ 
			R   \> $0,0$   \> $-1,1$ \> $1,-1$\\
			P  \> $1,-1$  \> $0,0$ \> $-1,1$\\
			S \> $-1,1$  \> $1,-1$ \> $0,0$
		\end{game}
		\label{fig:subRPS:RPS}
		}\hspace{5mm}
		\subfloat[A PAT of R-P-S]{
		\begin{game}{3}{3}
			\> R    \> P\> S\\
			R   \> $-5,8$  \> $-6,10$ \> $10,6$\\
			P  \> $-1,-5$  \> $-2,-3$ \> $2,-1$\\
			S \> $-9,13$  \> $2,9$ \> $6,11$
		\end{game}
		\label{fig:subPAT:RPS}
	    }\hspace{5mm}
		\subfloat[Zero-Sum Game equivalent to \protect\subref{fig:subPAT:RPS}]{
		\begin{game}{3}{3}
			\> R    \> P\> S\\
			R   \> $-16,16$  \> $-20,20$ \> $-12,12$\\
			P  \> $-12,12$  \> $-16,16$ \> $-20,20$\\
			S \> $-20,20$  \> $-12,12$ \> $-16,16$
		\end{game}
		\label{fig:subSZSG:RPS}
        }%
        \caption{\protect\subref{fig:subRPS:RPS} The classic zero-sum game Rock-Paper-Scissors. \protect\subref{fig:subPAT:RPS} A nonzero-sum game that is strategically equivalent to Rock-Paper-Scissors through a PAT. \protect\subref{fig:subSZSG:RPS} A zero-sum game that is strategically equivalent to the PAT of Rock-Paper-Scissors.}
        \label{fig:RPS}
        \end{figure}

By letting $\vec{w}=[-1, 1, 0]^\transpose$, $\vec{z}=[-1, 0, 1]^\transpose$, we have $\gamma=-\frac{\vec{w}^\transpose{A}\vec{z}}{\vec{w}^\transpose{B}\vec{z}}=-\frac{-12}{6}=2$. Then let $(i,j)=(1,1)$ to obtain the strategically zero-sum game:
\begin{align*}
    \bar{A}&={A}-\vec{1}_m\begin{bmatrix}11&14&22\end{bmatrix},\qquad \bar{B}=2{B}-\begin{bmatrix}\phantom{-}0\\-22\\\phantom{-}6\end{bmatrix}\vec{1}_n^\transpose.
\end{align*}
The result of these calculations is the zero-sum game $(m,n,\bar{A},\bar{B})$ which is displayed in Figure \ref{fig:subSZSG:RPS} and, as expected, has the NE strategies $\vec{p}^*=\vec{q}^*=[\frac{1}{3}, \frac{1}{3}, \frac{1}{3}]^\transpose$.

\subsection{Algorithm and Simulations}


We have shown that given the game $(m,n,{A},{B})$, one can determine if the game is strategically equivalent to the zero sum game $(m,n,\bar A,\bar B)$ through a PAT. If so, then it is possible to construct a rank-0 game which is strategically equivalent to the original game. One can then efficiently solve the strategically equivalent zero-sum game via linear programming.  We state the key steps in our algorithm below and show that both the determination of strategic equivalence and the computation of the strategically equivalent zero-sum game can be done in time $\bigO{(mn)}$.

The analytical results and discussions throughout this paper apply to real bimatrix games, with $({A},{B})\in\Re^{m\times n} \times \Re^{m\times n}$. However, for computational reasons, when discussing the algorithmic implementations we focus on rational bimatrix games, with  $({A},{B})\in\Q^{m\times n} \times \Q^{m\times n}$.

\begin{algorithm}
	\caption{Condensed algorithm for solving a strategically rank-0 game}
	\label{alg:solveNEShort}
	\begin{algorithmic}[1] 
		\Procedure{ShortSER0}{${A},{B}$}
		\If{${A}$ and/or ${B}\in\mc{M}_{m\times n}(\Q)$}
		    \State Calculate pure strategy NE and \textbf{exit}
		\Else
		    \State $\gamma\gets-\frac{\vec{w}^\transpose{A}\vec{z}}{\vec{w}^\transpose{B}\vec{z}}$
		    \If{$\gamma<0$}
		        \State Not strategically equivalent via PAT. \textbf{exit} 
		    \Else
		        \State $M\gets{A}+\gamma{B}$
		        \If{$M$ not in $\mc{M}_{m\times n}(\Q)$}
		            \State Not strategically equivalent via PAT. \textbf{exit}
		        \Else 
		            \State choose $(i,j)\in(m\times n)$
		            \State $\bar{A}\gets{A} - \vec{1}_mM_{(i)}$
		            \State $\bar{B}\gets \gamma B-(M^{(j)}-\vec{1}_mm_{i,j})\vec{1}_n^\transpose$
		            \State Solve $(m,n,\bar{A},\bar{B})$ via LP
		        \EndIf
		      \EndIf  
		\EndIf
		\EndProcedure
	\end{algorithmic}
\end{algorithm}

\begin{theorem}\label{thm:shortSER0}
The \textsc{SER0} algorithm determines if a game $(m,n,{A},{B})$ is strategically equivalent to a rank-$0$ game and returns the strategically equivalent zero-sum game in time $\bigO{(mn)}$.
\end{theorem}
We now turn our attention to the computational complexity of Algorithm \ref{alg:solveNEShort}. First, testing whether or not a matrix is in $\mc{M}_{m\times n}(\mathbb R)$ is equivalent to implementing \eqref{eq:prop:testFinM} and then comparing two matrices $F$ and $\bar F$.  Both these operations take time $\bigO{(mn)}$. Next, for any matrix $F\in\mc{M}_{m\times n}(\mathbb R)$, calculating $\vec{w}^\transpose F \vec{z}$ takes time $\bigO{(m^2)}+\bigO{(n)}<\bigO{(mn)}$ for $m\leq n$. So, calculating $\gamma$ can be done in time $\bigO{(mn)}$ if one has candidate vectors $(\vec{w},\vec{z})$. Corollary \ref{cor:prop:testFinM:A} gives an algorithm for determining such $(\vec{w},\vec{z})$ that runs in time $\bigO{(mn)}$.

Forming the $D$ matrix takes $mn$ multiplications and $mn$ additions, and therefore has time $\bigO{(mn)}$. Finally, calculating $\bar{A},\bar{B}$ for the case $\rank{M}=2$ consists of scalar-matrix multiplication, vector outer product, and matrix subtraction. Therefore, it has time $\bigO{(mn)}$.

This shows that overall the algorithm can both identify whether a game is strategically equivalent to a zero-sum game through a PAT and, if so, can determine the equivalent game in time $\bigO{(mn)}$. 

\subsection{Numerical Results}
To evaluate the performance of the SER0 algorithm we ran the following experiment. We fixed $m$ at $m=1000$ and varied $n$ with $n\in[2000,10000]$. Each entry in the payoff matrices are uniformly distributed. We ran SER0 on $1000$ different game instances for each value of $n$. 


All experiments were conducted on a standard desktop computer running Windows 7 with $16$GB of RAM and an Intel Xeon E5-1603 processor with $4$ cores running at $2.8$GHz.
\begin{figure}[H]
\begin{center}
    \includegraphics[width=0.6\textwidth]{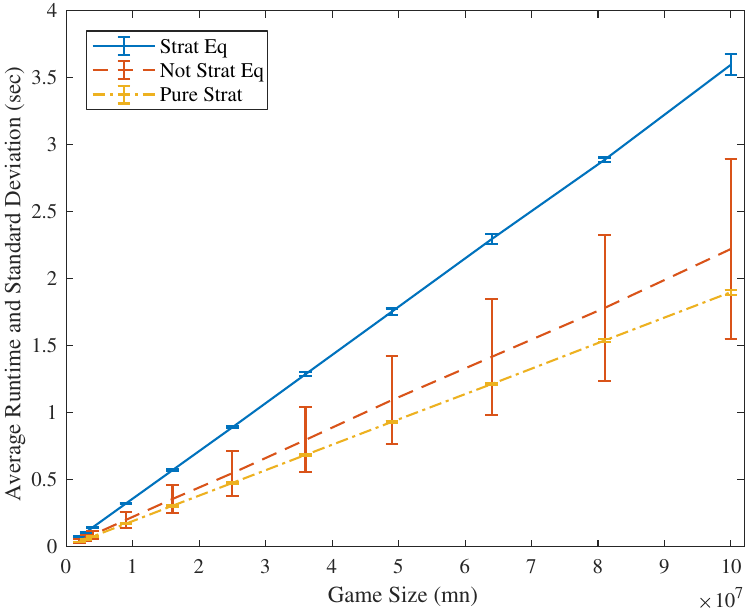}
	\caption{Average running time and standard deviation of the SER0 algorithm for strategically equivalent games, games that are not strategically equivalent, and games that are guaranteed to have at least one pure strategy NE. For each value of $mn$, we ran the algorithm on $1000$ such games.} \label{fig:runtimeRect}
\end{center}
\end{figure}
For each set of experiments, we created instances $(m,n,{A},{B})$ that were strategically equivalent to a zero-sum game $(m,n,\bar A,\bar B)$. In all cases tested, the SER0 algorithm correctly identified the games as strategically zero-sum and calculated the equivalent game $(m,n,\bar {A},\bar {B})$. As expected, one can observe from Figure \ref{fig:runtimeRect} that there is a clear linear relationship between the runtime of SER0 and the size of the game instance.  In addition, for very large games of size $(1000 \times 10000)$ SER0 found the equivalent game in an average time of $3.6$ seconds.

We then created games which were guaranteed to have a pure strategy NE. In other words, at least one of ${A}$ or ${B}$ were in $\mc{M}_{m\times n}(\Q)$. Again, the SER0 algorithm correctly identified all of these cases as having a pure strategy NE. Again, we observe that SER0's runtime is linear in this case. As expected, this case is much faster than the strategically equivalent case as there is no need to calculate $\bar{A},\bar{B},M$ nor test for $M\in\mc{M}_{m\times n}(\Q)$.

Finally, we conducted experiments on games that were not strategically equivalent to a zero-sum game via a PAT. Similar to the other two cases, the SER0 algorithm correctly identified these games as not strategically equivalent to a zero-sum game. As Figure \ref{fig:runtimeRect} show, this case also exhibits a linear relationship between runtime and the game size, although with a much higher standard deviation compared to the other two sets of experiments. This higher standard deviation is readily explainable by examining the SER0 algorithm.  When testing whether or not a game is strategically equivalent, the test can return a negative result if either $\gamma\leq0$ or $M\notin\mc{M}_{m\times n}(\Q)$. For cases of $\gamma\leq 0$, the algorithm terminates and returns a negative result in much less time than it takes to calculate the $M$ matrix and test for $M\in\mc{M}_{m\times n}(\Q)$. 

\section{Approximate Nash Equilibrium Through an Affine Transformation}\label{sec:affine}
Thus far we have introduced sufficient and necessary conditions of the existence of strategically equivalent zero-sum games via PAT. Then, we proposed an algorithm to compute a Nash equilibrium in time $\mathcal O(mn)$, given that the conditions in Theorem \ref{thm:existzsg} hold. Now, we introduce another algorithm that efficiently computes an approximate Nash equilibrium when conditions in Theorem \ref{thm:existzsg} are not satisfied.

There are two key insights that drive the second algorithm. Given a bimatrix game $(m, n, A, B)$, we can determine a strategically equivalent (not necessarily zero-sum) game $(m,n,R,C)$ via PAT. Let $E = R+C$. Consider the zero-sum game $(m, n, \tilde A, \tilde B)$, where $\tilde A = R-\frac{1}{2}E$ and $\tilde B = C - \frac{1}{2}E$. One may want the zero-sum game $(m, n, \tilde A, \tilde B)$ to be, in some sense, `close' to $(m,n,R,C)$. One approach is to determine matrices $R$ and $C$ to reduce the max norm of $E$, by which the error of the approximate NE can be controlled. 

The second insight is about the nature of the game transformation itself. From Definition \ref{def:strateEq}, we conclude that when the two games are strategically equivalent, then for every player, the best response correspondence is the same in both games. We identify an affine game transformations below that bijectively transforms the best response correspondence. In other words, the best response correspondence of each player in the transformed game can be bijectively mapped to the best response correspondence of the original game.

\begin{theorem}\label{thm:bestResponseMap}
Let $D_1\in\mathcal D_n$ and $D_2\in\mathcal D_m$ be positive definite diagonal matrices and cosndier the two games $(m,n,A,B)$ and $(m,n,AD_1,D_2B)$. We have
\begin{equation*}
    \Gamma_A(\vec q) = \Gamma_{AD_1}\left( \frac{D_1^{-1} \vec q}{\vec{1}_n^\transpose(D_1^{-1} \vec q)}\right), \qquad \Gamma_B(\vec p) = \Gamma_{D_2B}\left(\frac{D_2^{-1} \vec p}{\vec{1}_m^\transpose(D_2^{-1} \vec p)}\right).
\end{equation*}
\end{theorem}
\begin{proof}
Since $D_1$ is positive definite diagonal matrix, we have $\vec 1_n^\transpose(D_1^{-1}\vec q)>0$. This yields
\begin{align*}
    \vec p \in \Gamma_A(\vec q) & \iff \vec p^\transpose A \vec q = \max_i [A\vec q]_i, \\
    & \iff \frac{\vec p^\transpose A D_1 D_1^{-1} \vec q}{\vec 1_n^\transpose(D_1^{-1}\vec q)} = \max_i \left[\frac{A D_1 D_1^{-1} \vec q}{\vec 1_n^\transpose(D_1^{-1}\vec q)}\right]_i, \\
    & \iff \vec p \in \Gamma_{AD_1}\left(\frac{D_1^{-1}\vec q}{\vec 1_n^\transpose(D_1^{-1}\vec q)}\right).
\end{align*}
Hence the first equation is proved. The proof of the second equation is similar and therefore omitted.
\end{proof}

While the affine game transformation described above is simple and a straightforward extension of PAT, we are unable to find any reference in the literature on such an affine transformation. We next show that the Nash equilibria and approximate Nash equilibria of the transformed game can be mapped to that of the original game.
\begin{lemma}\label{lem:affineTransformation}
Consider games $(m,n,A,B)$ and $(m,n,AD_1,D_2B)$, where $D_1$, $D_2$ are diagonal, positive definite matrices. Let $\tilde{\vec p}^* = \frac{D_2^{-1} \vec p^*}{\vec{1}_m^\transpose(D_2^{-1} \vec p^*)}$, $\tilde{\vec q}^* = \frac{D_1^{-1} \vec q^*}{\vec{1}_n^\transpose(D_1^{-1} \vec q^*)}$. Then, \begin{enumerate}
    \item[(a)]$(\vec p^*, \vec q^*)$ is an $\varepsilon$-WSNE of the game $(m, n, A, B)$ if and only if the strategy pair $(\tilde{\vec p}^*, \tilde{\vec q}^*)$ is an $\tilde \varepsilon$-WSNE of the game $(m, n, AD_1, D_2B)$, where ${\varepsilon} = \tilde \varepsilon/\max\{\|D_1\|_{max}, \|D_2\|_{max}\}$
    \item[(b)]$(\vec p^*, \vec q^*)$ is a NE of the game $(m, n, A, B)$ if and only if the strategy pair $(\tilde{\vec p}^*, \tilde{\vec q}^*)$ is a NE of the game $(m, n, AD_1, D_2B)$.
\end{enumerate}
\end{lemma}
\begin{proof}
We start by proving (a). Suppose $(\tilde {\vec p}^*, \tilde{\vec q}^*)$ is an $\tilde \varepsilon$-WSNE of the game $(m,n,AD_1,D_2B)$. From Definition \ref{def:WSNE} we have 
\begin{align*}
    \forall i, k \in S_1, \tilde {\vec p}_i^* > 0 \implies (AD_1\tilde{\vec q}^*)_i \geq (AD_1\tilde{\vec q}^*)_k - \tilde \varepsilon.
\end{align*}
Denote $d_{2,j}$ as the $(j, j)th$ entry of matrix $D_2$. Since $d_{2, j} > 0 $ $\forall j \in S_1$, we have $\tilde{\vec p}_i^* > 0 \iff \left(\frac{D_2\tilde{\vec p}^*}{\vec{1}_m^\transpose(D_2\tilde{\vec p}^*)}\right)_i = \tilde {\vec p}^*_i > 0$. This yields
\begin{align*}
    (AD_1\tilde{\vec q}^*)_i \geq (AD_1& \tilde{\vec q}^*)_k - \tilde \varepsilon \\
    &\implies\bigg{(}\frac{AD_1D_1^{-1}\vec q^*}{\vec{1}_n^\transpose(D_1^{-1} \vec q^*)}\bigg{)}_i \geq \bigg{(}\frac{AD_1D_1^{-1}\vec q^*}{\vec{1}_n^\transpose(D_1^{-1} \vec q^*)}\bigg{)}_k - \frac{\tilde \varepsilon}{\vec{1}_n^\transpose(D_1^{-1} \vec q^*)}, \\
    &\iff (AD_1D_1^{-1}\vec q^*)_i \geq (AD_1D_1^{-1}\vec q^*)_k - \varepsilon,\\
    &\iff (A\vec q^*)_i \geq (A\vec q^*)_k - \varepsilon.
\end{align*}
The proof of the other player is similar and thus omitted. The proof of (b) is straightforward by setting $\varepsilon = 0$.
\end{proof}

\subsection{Approximate Nash Equilibrium}\label{sub:epsilonNE}
We have shown in Lemma \ref{lem:2epsWSNE} that additive transformations only affect $\varepsilon$ of the set of WSNE. Lemma \ref{lem:affineTransformation} shows the set of NE is preserved via affine transformations. Now we propose a theorem that establishes the connection (in terms of approximate WSNE) between a bimatrix game and the approximated zero-sum game.

\begin{theorem}\label{thm:epsilonSEG}
Let $D_1\in \mathcal D_n$, $D_2\in\mathcal D_m$, $\vec u\in\Re^n$, $\vec v\in\Re^m$. Consider games $(m,n,A,B)$ and $(m, n, R-\frac{1}{2}E, C-\frac{1}{2}E)$, where $R = AD_1-\vec 1_m\vec u^T$, $C = D_2B-\vec v\vec 1_n^T$, $E = R + C$, and define $\tilde \varepsilon := \|E\|_{max}$. Let 
\begin{equation*}
\vec p^* = \frac{D_2 \vec {\tilde p^*}}{\vec{1}_m^\transpose(D_2 \vec{ \tilde p^*})}, \qquad \vec q^* = \frac{D_1 \vec {\tilde q^*}}{\vec{1}_n^\transpose(D_1 \tilde{\vec q^*})}, \qquad \varepsilon =  \frac{\tilde\varepsilon}{\max\{\|D_1\|_{max}, \|D_2\|_{max}\}}.
\end{equation*}
If $(\tilde{\vec p}^*, \tilde{\vec q}^*)$ is a saddle point equilibrium of the zero sum game $(m, n, R-\frac{1}{2}E, C-\frac{1}{2}E)$, then the strategy pair $(\vec p^*, \vec q^*)$ is a $\varepsilon$-WSNE of the game $(m, n, A, B)$.
\end{theorem}

\begin{proof}
Suppose $(\vec {\tilde p^*}, \vec{\tilde q^*})$ is a saddle point equilibrium of the game $(m, n, R-\frac{1}{2}E, C-\frac{1}{2}E)$. From Lemma \ref{lem:2epsWSNE} we have $(\tilde{\vec p}^*, \tilde{\vec q}^*)$ is an $\tilde\varepsilon$-WSNE of game $(m, n, R, C)$. From Lemma \ref{lem:stratEqWSNE_PAT}, $(\tilde{\vec p}^*, \tilde{\vec q}^*)$ is an $\tilde \varepsilon$-WSNE of game $(m, n, AD_1, D_2B)$. Finally, Lemma \ref{lem:affineTransformation} implies that $(\vec p^*, \vec q^*)$ is a $\varepsilon$-WSNE of the original game $(m,n,A,B)$.
\end{proof}

\subsection{Algorithmic Implementation}\label{sec:algorithm}
In the sequel, we formulate an optimization problem to compute $D_1\in \mathcal D_n$, $D_2\in\mathcal D_m$, $\vec u\in\Re^n$, $\vec v\in\Re^m$ given $(A,B)$. This problem is equivalent to minimizing the max norm of the matrix $E=R+C=AD_1-\vec{1}_m\vec{u}^\transpose+D_2 B-\vec{v}\vec{1}_n^\transpose$, which is formulated as the optimization problem:
\begin{equation*}\label{prog:CP1}
\begin{aligned}
& \underset{\vec{u},\vec{v},D_1,D_2}{\text{min}}
& & \norm{AD_1-\vec{1}_m\vec{u}^\transpose+D_2 B-\vec{v}\vec{1}_n^\transpose}_{max} \\
& \text{s.t.}
& & \vec{u}\in\Re^n, \ \vec{v}\in\Re^m,D_1 \succeq I_n, D_2\succeq I_m, D_1 \in \ALP D_n, D_2\in \ALP D_m.\\
\end{aligned}\tag{CP1}
\end{equation*}
It is well known (see, for example \cite[pg.~150]{boyd2004convex}) that problems similar to \ref{prog:CP1} can be equivalently written in the epigraph form as:
\begin{equation*}\label{prog:LP1}
\begin{aligned}
& \underset{\vec{u},\vec{v},D_1, D_2,t}{\text{min}} 
& & \phantom{-}t\\
& \text{s.t.}
& & -d_{1,j} a_{i,j}- d_{2,i} b_{i,j} - t +u_j+v_i\leq 0, \qquad i=1,\dots, m, j=1, \dots, n \\ 
& & & d_{1,j} a_{i,j}+d_{2,i} b_{i,j} - t -u_j-v_i\leq 0, \qquad \quad\; i=1,\dots, m, j=1, \dots, n\\ 
& & & \phantom{-}\vec{u}\in\Re^n, \ \vec{v}\in\Re^m,d_{1,j}\geq 1,d_{2,i}\geq 1,t\geq 0.\\
\end{aligned}\tag{LP1}
\end{equation*}
The program in \ref{prog:LP1} is a linear program and can be solved in polynomial time. Thus, what we have shown is that given a nonzero-sum game $(m, n, A, B)$, we can, in polynomial time, find the zero-sum game $(m, n, R-\frac{1}{2}E, C-\frac{1}{2}E)$ such that $\tilde \varepsilon = \|E\|_{max}$ is minimized. Then, by Theorem \ref{thm:epsilonSEG}, we can compute an $\varepsilon$-WSNE of game $(m, n, A, B)$ by calculating the saddle point equilibrium $(\tilde{\vec p}^*, \tilde{\vec q}^*)$ of the zero-sum game. As a result, the algorithm finds an $\varepsilon$-WSNE of $(m, n, A, B)$ in polynomial time using two calls of a linear program. The algorithm is shown in Algorithm \ref{alg:solveEpsNE}.

\begin{algorithm}
	\caption{Algorithm for computing $\varepsilon$-Nash equilibrium of the game $(m,n,A,B)$.}
	\label{alg:solveEpsNE}
	\begin{algorithmic}[1] 
		\Procedure{ApproximateNE}{$m,n,A,B$}
		\State Solve \ref{prog:LP1} to get $D_1, D_2,\vec u,\vec v, \tilde \varepsilon$\;
		\State $R \gets AD_1-\vec{1}_m\vec{u}^\transpose$, $C \gets D_2 B-\vec{v}\vec{1}_n^\transpose$, $E \gets R+C$\;
		\State Calculate saddle point equilibrium $(\tilde {\vec p}^*, \tilde {\vec q}^*)$ of $(m, n, R-\frac{1}{2}E,C-\frac{1}{2}E)$\;
		\State 	Set $\vec p^* \gets \frac{D_2 \tilde{\vec p}^*}{\vec{1}_m^\transpose(D_2\tilde{\vec p}^*)}$, $\vec q^* \gets \frac{D_1\tilde{\vec q}^*}{\vec{1}_n^\transpose(D_1 \tilde{\vec q}^*)}$, $\varepsilon \gets \tilde \varepsilon/\max\{\|D_1\|_{max}, \|D_2\|_{max}\}$
		\EndProcedure
	\end{algorithmic}
\end{algorithm}

\subsection{Numerical Simulation}\label{sub:simulation}
To evaluate the performance of Algorithm \ref{alg:solveEpsNE}, we ran an experiment to evaluate the theoretical and actual error $\varepsilon$. We generated games with square payoff matrices with $n\in[5, 50]$ and uniformly distributed payoff values. We generated 10,000 different pairs of payoff matrices on each value of $n$.

For each game $(m, n, A, B)$, we first ran Algorithm \ref{alg:solveEpsNE} to get an approximate WSNE $(\vec {\tilde p}, \vec {\tilde q})$. Then, we ran Lemke-Howson Algorithm to compute the exact Nash equilibrium $(\vec p, \vec q)$. Finally, we computed the exact and theoretical $\varepsilon$. For each game size, we calculated the mean and variance of $\varepsilon$. In all cases tested, the actual error of $(\tilde{\vec p}, \tilde{\vec q})$ is less than the theoretical $\varepsilon$. The results are shown in Figure \ref{fig:epsmeanandvariance}.

\begin{figure}[ht]
  \centering
  \begin{minipage}[b]{0.49\textwidth}
    \includegraphics[width=\textwidth]{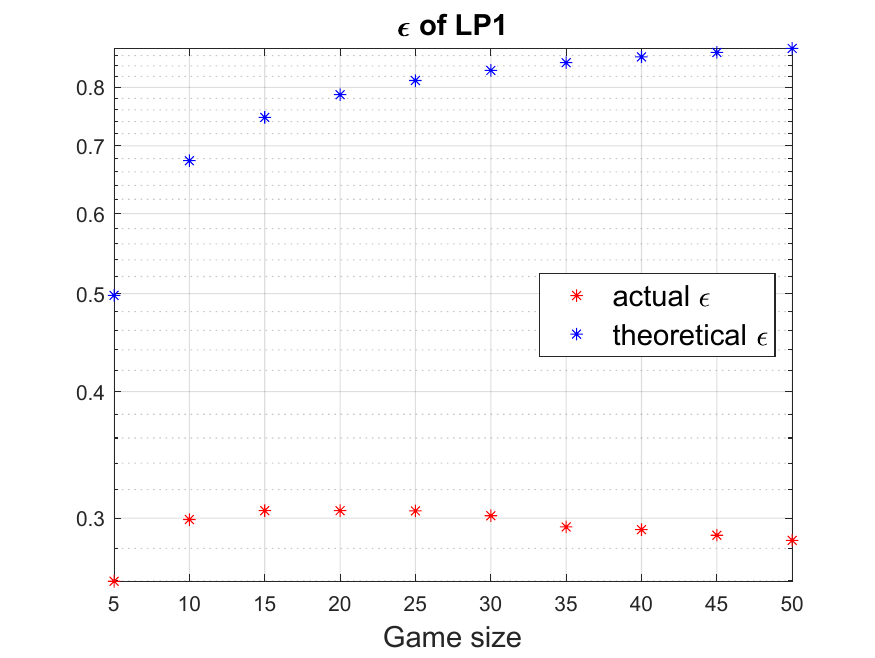}
    \caption{Actual and theoretical $\varepsilon$ with game sizes 5, 10, ..., 50}
  \end{minipage}
  \hfill
  \begin{minipage}[b]{0.49\textwidth}
    \includegraphics[width=\textwidth]{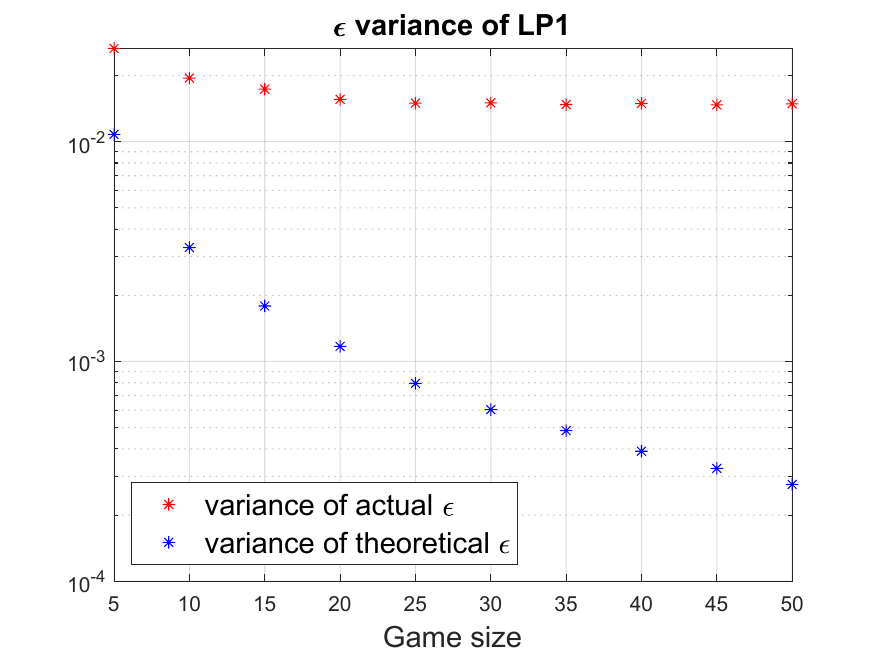}
    \caption{Actual and theoretical variance of $\varepsilon$ with game sizes 5, 10, ..., 50}
  \end{minipage}
  \caption{The results of Algorithm \ref{alg:solveEpsNE}.}\label{fig:epsmeanandvariance}
\end{figure}

\section{Conclusion}\label{sec:conclusion}
In this paper, we proposed two algorithms to determine a strategically equivalent zero-sum game given a non-zero-sum game. In this process, we proposed a new result on best response bijection, by which we can compute the NE of a bimatrix game by computation of the fixed point of the transformed function. The algorithms to determine the strategically equivalent zero-sum games run in polynomial time in the size of the game. Consequently, we showed that a class of non-zero sum games can be solved in polynomial time. 

In the cases where our second algorithm does not output a strategically equivalent zero-sum game, we show that it is a new algorithm for computing an $\varepsilon$-WSNE in polynomial time. Finally, we conducted numerical studies to show the efficacy of our algorithms.

\appendix
\section{Some Auxiliary Results on $\mathcal M_{m\times n}(\Re)$}\label{sec:auxillary}
This appendix is based on Appendix A.1.1 of the second author's PhD dissertation \cite{heyman2019computation}.
\begin{theorem}\label{thm:rankreduction}
For matrix $M \in \mathbb R^{m\times n}$ with $\rank{M} = r$. Let $M_1 = M$, and \begin{align}\label{eq:rankReductionRecursion}
    M_{k+1} := M_k - w_k^{-1}M_k\vec x_k \vec y_k^\transpose M_k, \qquad k\in \{1,\dots,r\}.
\end{align}
Define rank 1 matrices $W_k := \vec v_k \vec u_k^\transpose  = w_k^{-1}M_k\vec x_k \vec y_k^\transpose M_k$. Then, for any $j > k$, $\vec v_k \not \in \colspan{M_j}$ and $\vec u_k \not \in \colspan{M_j^\transpose}$.
\end{theorem}
\begin{proof}
By (\ref{eq:rankReductionRecursion}), we have $M_{k+1} = M - \sum_{i=1}^{k}W_i = \sum_{i = k+1}^r W_i$. Then,  \cite[p. 69]{wedderburn1934lectures} implies $\rank{M_{k+1}} = \rank{M_k} - 1$. Hence the result is proved.
\end{proof}

\begin{theorem}\label{thm:wznonempty}
If a matrix $A \in \mathbb R^{m\times n}$ and $A \not \in \mathcal M_{m\times n}(\mathbb R)$, then $\mathcal{WZ}(A) \neq \emptyset$
\end{theorem}
\begin{proof}
We can write $A = M + \sum_{i=1}^k \vec v_i \vec u_i^\transpose$, where \begin{enumerate}
    \item[(a)] $k = \rank{A} - \rank{M}$,
    \item[(b)] $M = \vec 1_m\vec u_o^\transpose + \vec v_0 \vec 1_n^\transpose \in \mathcal M_{m\times n}(\mathbb R)$, so $\rank{M} \leq 2$,
    \item[(c)] For any $i \in \{1, \dots, k\}$, $\vec v_i \not \in \colspan M$ and $\vec u_i \not \in \colspan {M^\transpose}$,
    \item[(d)] $\left \{\vec v_i\right \}_{i = 1}^k$ and $\left \{\vec u_i\right \}_{i = 1}^k$ are linearly independent.
\end{enumerate} 
Let $\vec w_0 = \vec 1_m$ and $\vec z_0 = \vec 1_n$, and construct orthogonal vectors such that
\begin{align*}
    &\vec w_i = \vec v_i - \sum_{j=0}^i \frac{\vec v_i^\transpose\vec w_j}{\vec w_j^\transpose \vec w_j}\vec w_j  &\forall i \in \{1,\dots,k\}\\
    &\vec z_i = \vec u_i - \sum_{j=0}^i \frac{\vec u_i^\transpose\vec z_j}{\vec z_j^\transpose \vec z_j}\vec z_j  &\forall i \in \{1,\dots,k\}\\
    &\vec 1_m^\transpose \vec w_i = \vec 1_n^\transpose \vec z_i = 0  &\forall i \in \{1,\dots,k\} \\
    &\vec w_i^\transpose\vec v_j = \vec u_j^\transpose\vec z_i = 0  &\forall j < i \\
    &\vec w_i^\transpose\vec v_i \neq 0 \text{ and } \vec u_i^\transpose\vec z_i \neq 0 &\forall i \in \{1,\dots,k\}
\end{align*}
After the iteration, we have that 
\begin{align*}
    \vec w_k^\transpose A \vec z_k &= \vec w_k^\transpose M \vec z_k + \vec w_k^\transpose \sum_{i=1}^{k} \vec v_i \vec u_i^\transpose \vec z_k = \vec w_k^\transpose M \vec z_k + \vec w_k^\transpose \sum_{i=1}^{k-1} \vec v_i \vec u_i^\transpose \vec z_k + \vec w_k^\transpose\vec v_k \vec u_k^\transpose \vec z_k \\
    &= \vec w_k^\transpose\vec v_k \vec u_k^\transpose \vec z_k \neq 0.
\end{align*}
The last step is because $ \vec w_i^\transpose\vec v_j = \vec u_j^\transpose\vec z_i = 0  $ for any $j < i$.
\end{proof}
\begin{lemma}\label{lem:M}
For any matrix $M\in\mc{M}_{m\times n}(\Re)$:
\begin{enumerate}
    \item If $\rank{M}=2$, then $\vec{1}_m\in\colspan{M}$ and $\vec{1}_n\in\colspan{M^\transpose}$. In addition, for all $\vec{x},\vec{y}$ such that $M\vec{x}=\vec{1}_m$,$M^\transpose \vec{y}=\vec{1}_n$,$\vec{1}_n^\transpose \vec{x}=0$,$\vec{1}_m^\transpose \vec{y}=0$.\label{itm:lem:Mr2}
    \item If $\rank{M}=1$, then either $\vec{1}_m\in\colspan{M}$, or $\vec{1}_n\in\colspan{M^\transpose}$ or both $\vec{1}_m\in\colspan{M}$ and $\vec{1}_n\in\colspan{M^\transpose}$.\label{itm:lem:Mr1}
\end{enumerate}
\end{lemma}
\begin{proof}
For claim \ref{itm:lem:Mr2}, $\vec{1}_m\in\colspan{M}$ and $\vec{1}_n\in\colspan{M^\transpose}$ follows directly from $\rank{M}=2$. In addition, $M\in\mc{M}_{m\times n}(\Re)$ and $\rank{M}=2$ implies that there exists $\vec{v}\neq\vec{0}_n$ and $\vec{u}\neq\vec{0}_m$ such that $M=\vct{1}_m\vct{u}^\transpose+\vct{v}\vct{1}_n^\transpose$. Also since $\rank{M}=2$, we have that for all $a\in\Re$, $\vec{v}\neq a\vec{1}_m$ since $\vec{v}$ and $\vec{1}_m$ must be linearly independent. Then, for all $\vec{x}$ such that $M\vec{x}=\vec{1}_m$ we have that:
\begin{align*}
    M\vec{x}&=\vec{1}_m\vec{u}^\transpose \vec{x}+\vec{v}\vec{1}_n^\transpose\vec{x}=\vec{1}_m=(\vec{u}^\transpose \vec{x})\vec{1}_m+(\vec{1}_n^\transpose\vec{x})\vec{v}=\vec{1}_m.
\end{align*}
This implies $(\vec{1}_n^\transpose\vec{x})\vec{v}=(1-\vec{u}^\transpose \vec{x})\vec{1}_m$. Further, $\vec{v}\neq a\vec{1}_m$ implies that the equation above is satisfied if and only if $\vec{1}_n^\transpose\vec{x}=0$ and $\vec{u}^\transpose \vec{x}=1$. To prove that for all $\vec{y}$ such that $M^\transpose \vec{y}=\vec{1}_n$, $\vec{1}_m^\transpose \vec{y}=0$ apply the same technique to $M^\transpose$.

Claim \ref{itm:lem:Mr1} follows directly from $\rank{M}=1$.
\end{proof}

\begin{lemma}\label{lem:FinMis0}
For any matrix $F\in\Re^{m\times n}$, if there exists $i,j$ such that $F_{(i)}=\vec{0}_n^\transpose$ and $F^{(j)}=\vec{0}_m$ then $F\in\mc{M}_{m\times n}(\Re)$ if and only if $F=\vec{0}_{m\times n}$.
\end{lemma}
\begin{proof}
Clearly $F=\vec{0}_{m\times n}$ implies that $F\in\mc{M}_{m\times n}(\Re)$ and that for all $i,j$ $F_{(i)}=\vec{0}_n^\transpose$ and $F^{(j)}=\vec{0}_m$. Now, consider the forward direction and suppose that $F\in\mc{M}_{m\times n}(\Re)$. Then from the definition of $\mc{M}_{m\times n}(\Re)$, we have that $\rank{F}\leq2$.  We will show that $\rank{F}=0$. $F_{(i)}=\vec{0}_n^\transpose$ implies that $\vec{1}_m\notin\colspan{M}$ and $F^{(j)}=\vec{0}_m$ implies that $\vec{1}_n\notin\colspan{M^\transpose}$. Then, by Lemma \ref{lem:M} $\rank{F}\neq 1,2$. Therefore, $\rank{F}=0$ and $F=\vec{0}_{m\times n}$.  
\end{proof}
\bibliography{ganesec2021_v3}
\bibliographystyle{splncs04}
\end{document}